\documentclass[11pt]{amsart}
\usepackage{amsmath}
\usepackage{amsfonts}
\usepackage{amssymb}
\usepackage{enumerate}
\usepackage{graphicx}
\usepackage{hyperref}

\newtheorem{thm}{Theorem}

\newtheorem{cor}[thm]{Corollary}

\newtheorem{remark}[thm]{Remark}
\newtheorem{lemma}[thm]{Lemma}
\newtheorem{prop}[thm]{Proposition}

\newtheorem{exam}[thm]{Example}

\newcommand{\bra}[1]{\langle #1 |}
\newcommand{\ket}[1]{| #1 \rangle}
\newcommand{\braket}[2]{\langle #1 | #2 \rangle}
\newcommand{\ketbra}[2]{| #1 \rangle\langle #2 |}
\newcommand{\Tr}{{\rm Tr}}
\newcommand{\bb}[1]{\mathbb{#1}}
\newcommand{\cl}[1]{\mathcal{#1}}

\newcommand{\mnorm}[1]{%
\left\vert\kern-0.9pt\left\vert\kern-0.9pt\left\vert #1
\right\vert\kern-0.9pt\right\vert\kern-0.9pt\right\vert}

\newcommand{\bigmnorm}[1]{%
\big\vert\kern-0.9pt\big\vert\kern-0.9pt\big\vert #1
\big\vert\kern-0.9pt\big\vert\kern-0.9pt\big\vert}

\title[Quantum Gate Fidelity in Terms of Choi Matrices]{Quantum Gate Fidelity in Terms of Choi Matrices}

\author[N.~Johnston]{Nathaniel Johnston}
\address{Department of Mathematics and Statistics, University of Guelph,
Guelph, Ontario N1G 2W1, Canada}
\email{njohns01@uoguelph.ca}

\author[D.~W.~Kribs]{David W. Kribs}
\address{Department of Mathematics and Statistics, University of Guelph, Guelph, Ontario N1G 2W1, Canada and Institute for Quantum Computing, University of Waterloo, Waterloo, Ontario N2L 3G1, Canada}
\email{dkribs@uoguelph.ca}

\begin{document}

\begin{abstract}
We provide new results for computing and comparing the quantum gate fidelity of quantum channels via their Choi matrices. We extend recent work that showed there exist non-dual pairs of quantum channels with equal gate fidelity by providing an explicit characterization of all such channels. We use our characterization to show that when the dimension is $2$ (or $3$, under slightly stronger hypotheses), the gate fidelity of two channels is equal if and only if their difference equals the difference of some unital map and its dual -- a fact that has been shown to be false when the dimension is $4$ or larger. We also present a formula for the minimum gate fidelity of a channel in terms of a well-studied norm on a compression of its Choi matrix. As a consequence, several new ways of bounding and approximating the minimum gate fidelity follow, including a simple semidefinite program to compute it for qubit channels.\medskip

\medskip

\end{abstract}

\maketitle

\section{Introduction}

In quantum information theory, many of the most important quantum operations are represented ideally by unitary transformations \cite{NC00}. Experimentally, however, gates are imperfectly implemented via trace-preserving, completely positive maps (called quantum channels). The hope is that the quantum channel which is implemented is in some sense ``close'' to the desired unitary channel. One of the most common techniques to measure the distance between quantum channels and unitary channels is via the quantum gate fidelity. The gate fidelity is a function, sending pure states to real numbers, that measures the amount of overlap between the output of the unitary channel and the output of the implemented quantum channel. The goal of the present paper is to characterize the gate fidelity in terms of the Choi matrix \cite{C75} of a given channel.

Recent work \cite{Mag11} has investigated under what conditions two different quantum channels can have the same gate fidelity. It was demonstrated in that there exist non-dual pairs of channels with the same gate fidelity in all dimensions $n \geq 4$. We extend this work by providing an easily-testable necessary and sufficient condition that described exactly when two channels have the same gate fidelity. We use our characterization to show that for qubit channels, the only way for two channels to have the same gate fidelity is if their difference is the scaled difference of a unital channel and its dual. This property still holds for channels on a $3$-dimensional system as long as the channels are unital, but it fails for higher-dimensional channels.

Our characterization of the gate fidelity is in terms of a compression of the Choi matrix of the channel. In addition to being useful for characterizing equality of the gate fidelity, we show that it can be used to calculate the average and variance of the gate fidelity. Furthermore, we show that the minimum gate fidelity can be phrased in terms of a certain norm \cite{JK10a,JK10b} on that operator, which allows us to bound the minimum gate fidelity in a variety of new ways and easily compute it for qubit channels.

The remainder of the paper is organized as follows. In Section~\ref{sec:channels} we introduce the mathematical basics of quantum channels, the gate fidelity, and the measures based on the gate fidelity that are most frequently used. In Section~\ref{sec:choi} we present the symmetric subspace, flip operator, and Choi matrices, which are the key ingredients in our characterization of the gate fidelity. Section~\ref{sec:characterizeEqual} contains our main result, which characterizes gate fidelity in terms of the channel's Choi matrix, and Section~\ref{sec:SmallDims} specializes our result to the case when the channel acts on a $2$- or $3$-dimensional system. We close in Section~\ref{sec:minGateFid} by exploring how our characterization applies to minimum gate fidelity and we demonstrate how to calculate it for qubit channels.

\section{Quantum Channels and Gate Fidelity}\label{sec:channels}

Throughout this work, we will denote a (finite-dimensional) complex Hilbert space of dimension $n$ by $\cl{H}_n$. The space of linear operators on $\cl{H}_n$ will be denoted $\cl{L}(\cl{H}_n)$. We will use $I$ to denote the identity operator on $\cl{H}_n$ and $id_n$ to denote the identity operator on $\cl{L}(\cl{H}_n)$. It will sometimes be useful to associate $\cl{H}_n$ with $\bb{C}^n$, and $\cl{L}(\cl{H}_n)$ with the space of $n \times n$ complex matrices via matrix representations of operators in a given orthonormal basis, so we will do so freely without making special mention of that association.

A \emph{pure quantum state} is represented by a unit vector $\ket{\phi} \in \cl{H}_n$. We will denote the standard basis of $\cl{H}_n$ by $\ket{0},\ket{1},\ldots,\ket{n-1}$. We will frequently work with bipartite Hilbert spaces $\cl{H}_n \otimes \cl{H}_n$, and we will make use of the shorthand notation $\ket{\phi\psi} := \ket{\phi} \otimes \ket{\psi} \in \cl{H}_n \otimes \cl{H}_n$. One particularly important bipartite pure state is the \emph{maximally entangled} state $\ket{\psi_+} := \frac{1}{\sqrt{n}}\sum_{j=0}^{n-1}\ket{jj}$.

Not all quantum states are pure however, and a general \emph{mixed quantum state} is represented by positive-semidefinite operator $\rho \in \cl{L}(\cl{H}_n)$ with $\Tr(\rho) = 1$, where $\Tr(\cdot)$ denotes the trace. Note that a pure state $\ket{\phi}$ can be represented by the operator $\ketbra{\phi}{\phi}$, where $\bra{\phi} := \ket{\phi}^\dagger$ is the dual vector of $\ket{\phi}$. The \emph{fidelity} \cite{Uhl76,Joz94} of two quantum states $\rho$ and $\sigma$ is defined by
\begin{align*}
	\cl{F}(\rho,\sigma) := \Tr\left(\sqrt{\sqrt{\rho}\sigma\sqrt{\rho}}\right)^2,
\end{align*}
which reduces in the case when $\sigma$ is a pure state to simply $\cl{F}(\rho,\ketbra{\phi}{\phi}) = \bra{\phi}\rho\ket{\phi}$. The fidelity can be thought of as a measure of how well $\rho$ and $\sigma$ can be distinguished, and it satisfies $0 \leq \cl{F}(\rho,\sigma) \leq 1$ with $\cl{F}(\rho,\sigma) = 1$ if and only if $\rho = \sigma$.

A \emph{quantum channel} is a completely positive, trace-preserving linear map $\cl{E} : \cl{L}(\cl{H}_n) \rightarrow \cl{L}(\cl{H}_n)$. Every quantum channel admits a family of \emph{Kraus operators} \cite{NC00} $\{E_i\}$ such that $\cl{E}(\rho) = \sum_i E_i \rho E_i^{\dagger}$ for all $\rho$ and $\sum_i E_i^{\dagger} E_i = I$. Given a channel $\cl{E} : \cl{L}(\cl{H}_n) \rightarrow \cl{L}(\cl{H}_n)$, its \emph{dual channel} $\cl{E}^{\dagger} : \cl{L}(\cl{H}_n) \rightarrow \cl{L}(\cl{H}_n)$ is defined via the Hilbert-Schmidt inner product to be the unique map such that $\Tr(\cl{E}(X)Y) = \Tr(X\cl{E}^{\dagger}(Y))$ for all $X,Y \in \cl{L}(\cl{H}_n)$. It is the case that $\cl{E}^{\dagger}$ is completely positive if and only if $\cl{E}$ is completely positive, and $\cl{E}$ is trace-preserving if and only if $\cl{E}^{\dagger}$ is unital -- both of these facts can be seen by noting by cyclicity of the trace shows that
\begin{align*}
	\Tr(\cl{E}(X)Y) = \Tr\left( \sum_i E_i X E_i^{\dagger}Y\right) = \Tr\left( X \sum_i E_i^{\dagger}Y E_i\right),
\end{align*}
which implies that if $\{E_i\}$ is a family of Kraus operators for $\cl{E}$ then $\{E_i^{\dagger}\}$ is a family of Kraus operators for $\cl{E}^{\dagger}$.

In the special case when a channel $\cl{U}$ satisfies $\cl{U}(\rho) = U\rho U^\dagger$ for some unitary operator $U$, $\cl{E}$ is called a \emph{unitary channel}. Unitary channels are exactly the channels that do not introduce mixedness (i.e., decoherence) into states and thus they very often are the types of channels that are meant to be implemented in experimental settings. However, no implementation of a channel is perfect -- errors are introduced that cause the channel that is implemented to not actually be unitary. The \emph{gate fidelity} is a tool for comparing how well the implemented quantum channel $\cl{E}$ approximates the desired unitary channel $\cl{U}$. Gate fidelity is a function defined on pure states as follows:
\begin{align*}
	\cl{F}_{\cl{E},\cl{U}}(\ket{\phi}) := \cl{F}(\cl{E}(\ketbra{\phi}{\phi}),\cl{U}(\ketbra{\phi}{\phi})) = \bra{\phi}U^\dagger \cl{E}(\ketbra{\phi}{\phi})U\ket{\phi}.
\end{align*}

Without loss of generality, we can assume $U = I$ by noting that
\begin{align*}
	\bra{\phi}U^\dagger \cl{E}(\ketbra{\phi}{\phi})U\ket{\phi} = \bra{\phi} (\cl{U}^\dagger \circ \cl{E})(\ketbra{\phi}{\phi})\ket{\phi},
\end{align*}
so $\cl{F}_{\cl{E},\cl{U}} = \cl{F}_{\cl{U}^\dagger \circ \cl{E},id_n}$. For brevity, we will use the shorthand $\cl{F}_{\cl{E}} := \cl{F}_{\cl{E},id_n}$, which can be thought of as measuring how noisy the channel $\cl{E}$ is. It will also be useful occasionally to consider the gate fidelity of linear maps that are not actually channels. That is, for \emph{any} linear map $\Lambda : \cl{L}(\cl{H}_n) \rightarrow \cl{L}(\cl{H}_n)$ we define $\cl{F}_{\Lambda}(\ket{\phi}) = \bra{\phi}\Lambda(\ketbra{\phi}{\phi})\ket{\phi}$.

The two most well-studied distance measures based on the gate fidelity are the \emph{average gate fidelity} $\overline{\cl{F}_{\cl{E}}}$ \cite{HHH99,N02,BOSBJ02,GLN05,EAZ05} and the \emph{minimum gate fidelity} \cite{NC00,GLN05}
\begin{align}
\cl{F}_{\cl{E}}^{min} = \min_{\ket{\phi}} \cl{F}_{\cl{E}} (\ket{\phi}),
\end{align}
which are obtained by either averaging (via the Fubini-Study measure \cite{BZ06}) or minimizing over all pure states $\ket{\phi}$, respectively. The minimum gate fidelity has the interpretation as the most noise that $\cl{E}$ can introduce into a quantum system. It makes sense then that one might want instead to minimize $\cl{F}(\cl{E}(\rho),\rho)$ over all mixed states $\rho$. The reason we minimize over pure states is that joint concavity of fidelity implies that minimizing over mixed states $\rho$ gives the exact same quantity $\cl{F}_{\cl{E}}^{min}$ as minimizing over pure states $\ket{\phi}$ -- see \cite[Section 9.3]{NC00} or \cite[Section IV.C]{GLN05} for a proof of this fact.

One of the most celebrated results concerning gate fidelity is an explicit formula for the average gate fidelity of a quantum channel in terms of its Kraus operators $\{E_i\}$ \cite{HHH99,N02}:
\begin{align}\label{eq:avg_fid}
	\overline{\cl{F}_{\cl{E}}} = \frac{n + \sum_i |\Tr(E_i)|^2}{n(n+1)}.
\end{align}
Similarly, higher-order moments of the gate fidelity have been computed \cite{PMM08,MBE09}. However, the minimum gate fidelity seems to be much more difficult to compute -- for some partial results and bounds on minimum gate fidelity, see \cite{KBO09,LRKKB11}. In Section~\ref{sec:minGateFid} we will derive a formula for the minimum gate fidelity that will allow us to efficiently compute it for qubit channels and derive new bounds for it in general.

\section{Choi Matrices and the Symmetric Subspace}\label{sec:choi}

To every linear map $\cl{E} : \cl{L}(\cl{H}_n) \rightarrow \cl{L}(\cl{H}_n)$ there is an associated \emph{Choi matrix}:
\[
C_{\cl{E}} = (id_n \otimes \cl{E})(n\ketbra{\psi_+}{\psi_+}).
\]
This identification between linear maps on $\cl{L}(\cl{H}_n)$ and operators in $\cl{L}(\cl{H}_n \otimes \cl{H}_n)$ is known as the \emph{Choi-Jamiolkowski isomorphism}. A celebrated result of Choi says that $\cl{E}$ is completely positive if and only if $C_{\cl{E}}$ is positive semidefinite \cite{C75}. Similarly, $\cl{E}$ is trace-preserving if and only if $\Tr_1(C_{\cl{E}}) = I$ and $\cl{E}$ is unital (i.e., $\cl{E}(I) = I$) if and only if $\Tr_2(C_{\cl{E}}) = I$, where $\Tr_j$ denotes the partial trace over the $j$th subsystem.

The \emph{flip operator} $F \in \cl{L}(\cl{H}_n) \otimes \cl{L}(\cl{H}_n)$ is the Choi matrix of the transpose map $T$. Its name comes from the fact that $F\ket{\phi\psi} = \ket{\psi\phi}$ for any $\ket{\phi},\ket{\psi} \in \cl{H}_n$. The \emph{symmetric subspace} $\cl{S} \subseteq \cl{H}_n \otimes \cl{H}_n$ is the subspace spanned by the states of the form $\ket{i}\ket{j} + \ket{j}\ket{i} \in \cl{H}_n \otimes \cl{H}_n$. The Takagi factorization \cite{Tag24,HJ85} of complex symmetric matrices (and hence symmetric states) says that $\ket{\phi} \in \cl{S}$ if and only if $\ket{\phi}$ has a symmetric Schmidt decomposition: $\ket{\phi} = \sum_{j=1}^n \alpha_j \ket{\phi_j \phi_j}$. We will denote the projection of $\cl{H}_n \otimes \cl{H}_n$ onto $\cl{S}$ by $P_{\cl{S}}$. Notice that $P_{\cl{S}} = \frac{1}{2}(I + F)$ and that the dimension of $\cl{S}$ is $n(n+1)/2$.

We now present a simple proposition concerning the average gate fidelity. While this proposition may not appear particularly useful considering we already have Equation~(\ref{eq:avg_fid}) to work with, it gives $\overline{\cl{F}_{\cl{E}}}$ in terms of the operator $P_{\cl{S}}(T \otimes id_n)(C_{\cl{E}})P_{\cl{S}}$. We will see this operator repeatedly throughout this work, most notably in a characterization of the gate fidelity and also in a formula for the minimum gate fidelity, and hence it is useful to see how it relates to the average gate fidelity as well. Observe that the scaling factor $2/(n(n+1))$ in the following result is exactly one divided by the dimension of $\cl{S}$, so average gate fidelity can be seen as an average over the symmetric subspace.

\begin{prop}\label{prop:avg_gate_fid}
	Let $\cl{E} : \cl{L}(\cl{H}_n) \rightarrow \cl{L}(\cl{H}_n)$ be a quantum channel. Then
	\begin{align*}
		\overline{\cl{F}_{\cl{E}}} = \frac{2}{n(n+1)}\Tr\big(P_{\cl{S}}(T \otimes id_n)(C_{\cl{E}})P_{\cl{S}}\big).
	\end{align*}
\end{prop}

\begin{proof}
	The proof is by straightforward algebra. If we write $C_{\cl{E}}$ in its spectral decomposition $\sum_i \lambda_i \ketbra{v_i}{v_i}$, then
	\begin{align*}
		\Tr\big((T \otimes id_n)(C_{\cl{E}})(2P_{\cl{S}} - I)\big) & = \Tr\big((T \otimes id_n)(C_{\cl{E}})F\big) \\
		& = \Tr\big((T \otimes id_n)(C_{\cl{E}})(id_n \otimes T)(n\ketbra{\psi_+}{\psi_+})\big) \\
		& = n\bra{\psi_+}C_{\cl{E}}\ket{\psi_+} \\
		& = n\sum_i \lambda_i |\braket{\psi_+}{v_i}|^2 \\
		& = \sum_i |\Tr(E_i)|^2,
	\end{align*}
	where the third equality follows from the identity $T^\dagger = T$, and the final equality comes from the fact that the (scaled) eigenvectors of $C_{\cl{E}}$ are the vectorizations of the Kraus operators of $\cl{E}$. The result follows from Equation~(\ref{eq:avg_fid}) and the fact that $\cl{E}$ is trace-preserving, so $\Tr((T \otimes id_n)(C_{\cl{E}})) = n$.
\end{proof}

In fact, the proof of Proposition~\ref{prop:avg_gate_fid} shows that $(2/n^2)\Tr\big(P_{\cl{S}}(T \otimes id_n)(C_{\cl{E}})P_{\cl{S}}\big) - 1/n = \bra{\psi_+}C_{\cl{E}}\ket{\psi_+}/n$, a quantity that was referred to as $\chi_{0,0}$ in \cite{MBE09}. It follows that the formulas for the variance and higher-order moments of $\cl{F}_{\cl{E}}$ can be written in terms of the operator $P_{\cl{S}}(T \otimes id_n)(C_{\cl{E}})P_{\cl{S}}$ as well.

For completeness, we close this section with a well-known fact that shows how the Choi matrices of a pair of dual channels $\cl{E}$ and $\cl{E}^{\dagger}$ are related to each other. For another proof in the more general case of positive (not necessarily completely positive) maps, see \cite[Lemma 3]{S11}.
\begin{lemma}\label{lem:dual_choi}
	Let $\cl{E} : \cl{L}(\cl{H}_n) \rightarrow \cl{L}(\cl{H}_n)$ be a quantum channel. Then $C_{\cl{E}^\dagger} = F C_{\cl{E}}^T F$.
\end{lemma}

\begin{proof}
	Use the spectral decomposition to write $C_{\cl{E}} = \sum_i \lambda_i \ketbra{v_i}{v_i}$. Then $C_{\cl{E}}^T = \sum_i \lambda_i \overline{\ketbra{v_i}{v_i}}$. It is easily verified that for any $X \in \cl{L}(\cl{H}_n)$, the vectorization of $X^T$ is exactly $F$ times the vectorization of $X$. The result follows from recalling that $\{E_i\}$ is a set of Kraus operators for $\cl{E}$ if and only if $\{E_i^{\dagger}\}$ is a set of Kraus operators for $\cl{E}^{\dagger}$, and that the Kraus operators of a channel can be chosen so that their vectorizations are the eigenvectors of $C_{\cl{E}}$.
\end{proof}

\section{Channels with Identical Gate Fidelity}\label{sec:characterizeEqual}

It was shown in \cite{Mag11} that if $n \geq 4$ and $\cl{E} : \cl{L}(\cl{H}_n) \rightarrow \cl{L}(\cl{H}_n)$ is a quantum channel with positive-definite Choi matrix, then there exists another quantum channel $\cl{R} : \cl{L}(\cl{H}_n) \rightarrow \cl{L}(\cl{H}_n)$ with $\cl{R} \neq \cl{E},\cl{E}^\dagger$ such that $\cl{F}_{\cl{E}} = \cl{F}_{\cl{R}}$. A particular consequence of this result is the fact that there exist non-depolarizing quantum channels with constant gate fidelity. In this section we expand on this work by providing an easily-testable characterization of when two maps have the same gate fidelity in terms of their Choi matrices.

We begin with a lemma that allows us to talk about $\cl{F}_{\cl{E}}(\ket{\phi})$ in terms of the Choi matrix of $\cl{E}$. This lemma is in essence a simplification of \cite[Lemma 1]{Mag11}, but we present it here for completeness, along with a simplified proof.
\begin{lemma}\label{lem:fidel_choi}
	Let $\Lambda : \cl{L}(\cl{H}_n) \rightarrow \cl{L}(\cl{H}_n)$ be a linear map and let $\ket{\phi} \in \cl{H}_n$. Then
	\begin{align*}
		\cl{F}_{\Lambda}(\ket{\phi}) = \bra{\phi\phi}(T \otimes id_n)(C_{\Lambda})\ket{\phi\phi}.
	\end{align*}	
\end{lemma}
\begin{proof}
	The proof is by simple algebra.
	\begin{align*}
		\bra{\phi\phi}(T \otimes id_n)(C_{\Lambda})\ket{\phi\phi} & = \sum_{i,j=0}^{n-1}\bra{\phi\phi}(T(\ketbra{i}{j}) \otimes \Lambda(\ketbra{i}{j}))\ket{\phi\phi} \\
		& = \sum_{i,j=0}^{n-1}\braket{\phi}{j}\braket{i}{\phi} \bra{\phi}\Lambda(\ketbra{i}{j})\ket{\phi} \\
		& = \bra{\phi}\Lambda\big((\sum_{i=0}^{n-1}\braket{i}{\phi}\ket{i}) (\sum_{j=0}^{n-1}\braket{\phi}{j} \bra{j})\big)\ket{\phi} \\
		& = \bra{\phi}\Lambda(\ketbra{\phi}{\phi})\ket{\phi} \\
		& = \cl{F}_{\Lambda}(\ket{\phi}).
	\end{align*}
\end{proof}

We are now ready to state the main result of this section, which allows us to determine whether or not two quantum channels have the same gate fidelity simply by comparing a certain modification of their Choi matrices. In particular, we see that the gate fidelity of a channel $\cl{E}$ is determined exactly by the operator $P_{\cl{S}}(T \otimes id_n)(C_{\cl{E}})P_{\cl{S}}$. Note that the ``if'' direction of this result follows trivially from \cite[Lemma 1]{Mag11}, so what is new here is the ``only if'' direction.

\begin{thm}\label{thm:equal_gate_choi}
	Let $Q, R : \cl{L}(\cl{H}_n) \rightarrow \cl{L}(\cl{H}_n)$ be linear maps. Then $\cl{F}_{Q} = \cl{F}_{R}$ if and only if
\[
P_{\cl{S}}(T \otimes id_n)(C_{Q})P_{\cl{S}} = P_{\cl{S}}(T \otimes id_n)(C_{R})P_{\cl{S}}.
\]
\end{thm}

\begin{proof}
	We begin by defining $D := (T \otimes id_n)(C_{Q} - C_{R})$. Using Lemma~\ref{lem:fidel_choi} with $\Lambda := Q - R$ shows that $\cl{F}_{Q} = \cl{F}_{R}$ if and only if
	\begin{align}\label{eq:sym_sep}
		\bra{\phi\phi}D\ket{\phi\phi} = 0 \quad \forall \, \ket{\phi} \in \cl{H}_n.
	\end{align}
	If we can show that Equation~(\ref{eq:sym_sep}) implies
	\begin{align*}
		\bra{s_1}D\ket{s_2} = 0 \quad \forall \, \ket{s_1},\ket{s_2} \in \cl{S},
	\end{align*}
	we will be done. To this end, fix arbitrary $\ket{\phi},\ket{\psi} \in \cl{H}_n$ and let $z \in \bb{C}$ be such that $|z| = 1$. Then Equation~(\ref{eq:sym_sep}) implies that
	\begin{align}\label{eq:z_sep_sym1}
		 (\bra{\phi}+\overline{z}\bra{\psi})\otimes(\bra{\phi}+\overline{z}\bra{\psi})D(\ket{\phi}+z\ket{\psi})\otimes (\ket{\phi}+z\ket{\psi}) & = 0 \quad \text{ and} \\\label{eq:z_sep_sym2}
		 (\bra{\phi}-\overline{z}\bra{\psi})\otimes(\bra{\phi}-\overline{z}\bra{\psi})D(\ket{\phi}-z\ket{\psi})\otimes (\ket{\phi}-z\ket{\psi}) & = 0.
	\end{align}
	If we expand and add Equations~(\ref{eq:z_sep_sym1}) and~(\ref{eq:z_sep_sym2}) together and use (\ref{eq:sym_sep}), we see that for all $z \in \bb{C}$ with  $|z| = 1$,
	\begin{align}\label{eq:sep_sym3}
		z^2\bra{\phi\phi}D\ket{\psi\psi} + \overline{z}^2\bra{\psi\psi}D\ket{\phi\phi} + (\bra{\phi\psi} + \bra{\psi\phi})D(\ket{\phi\psi} + \ket{\psi\phi}) = 0 .
	\end{align}
	If we subtract from Equation~(\ref{eq:sep_sym3}) the equation obtained by replacing $z$ in Equation~(\ref{eq:sep_sym3}) by $iz$, we learn that
	\begin{align*}
		z^2\bra{\phi\phi}D\ket{\psi\psi} + \overline{z}^2\bra{\psi\psi}D\ket{\phi\phi} = 0 \quad \forall z \in \bb{C} \text{ with } |z| = 1.
	\end{align*}
	The following two equations arise from letting $z = 1$ and $z = e^{i\pi/4}$, respectively:
	\begin{align*}
		\bra{\phi\phi}D\ket{\psi\psi} + \bra{\psi\psi}D\ket{\phi\phi} & = 0 \\
		i\bra{\phi\phi}D\ket{\psi\psi} - i\bra{\psi\psi}D\ket{\phi\phi} & = 0.
	\end{align*}
	Adding $i$ times the first equation to the second equation gives
	\begin{align}\label{eq:sym_sep_both}
		\bra{\phi\phi}D\ket{\psi\psi} = 0 \quad \forall \, \ket{\phi},\ket{\psi} \in \cl{H}_n.
	\end{align}
	
	Now let $\ket{s_1},\ket{s_2} \in \cl{S}$. By the Takagi factorization we know that we can write $\ket{s_1} = \sum_{j=1}^n\alpha_j\ket{\phi_j \phi_j}$ and $\ket{s_2} = \sum_{k=1}^n\beta_k\ket{\psi_k \psi_k}$. Thus
	\begin{align*}
		\bra{s_1}D\ket{s_2} & = \sum_{j,k=1}^n\alpha_j \beta_k \bra{\phi_j \phi_j}D\ket{\psi_k \psi_k} = 0,
	\end{align*}
	since each term in the sum equals zero by Equation~(\ref{eq:sym_sep_both}). This completes the proof.
\end{proof}

As a particularly important special case of Theorem~\ref{thm:equal_gate_choi}, consider the case of channels with constant gate fidelity. It is easily shown that any depolarizing channel has constant gate fidelity, and it was shown in~\cite{Mag11} that there exist non-depolarizing channels with constant gate fidelity. We now present a characterization of all channels with constant gate fidelity.

\begin{cor}\label{cor:const_gate_fid}
	Let $Q : \cl{L}(\cl{H}_n) \rightarrow \cl{L}(\cl{H}_n)$ be a linear map and let $c \in \bb{R}$. Then $\cl{F}_{Q} \equiv c$ if and only if $P_{\cl{S}}(T \otimes id_n)(C_{Q})P_{\cl{S}} = cP_{\cl{S}}$.
\end{cor}

\begin{proof}
	Consider the linear map $R : \cl{L}(\cl{H}_n) \rightarrow \cl{L}(\cl{H}_n)$ defined by $R(X) = \frac{1}{n-1}\big((cn - 1)X + (1-c)I\big)$ (if $1/n \leq c \leq 1$ then $R$ is a depolarizing quantum channel). Simple algebra reveals that
	\begin{align*}
		\cl{F}_{R}(\ket{\phi}) = \frac{1}{n-1}\bra{\phi}\big((cn - 1)\ketbra{\phi}{\phi} + (1-c)I\big)\ket{\phi} = \frac{(cn - 1) + (1 - c)}{n-1} = c
	\end{align*}
	for any $\ket{\phi} \in \cl{H}_n$. Also, $(T \otimes id_n)(C_{R}) = \frac{1}{n-1}\big((cn - 1)F + (1-c)(I \otimes I)\big)$. Thus,
\[
P_{\cl{S}}(T \otimes id_n)(C_{R})P_{\cl{S}} = \frac{1}{n-1}\big((cn - 1)P_{\cl{S}} + (1-c)P_{\cl{S}}\big) = cP_{\cl{S}}.
\]
Using Theorem~\ref{thm:equal_gate_choi} gives the result.
\end{proof}

Corollary~\ref{cor:const_gate_fid} can be seen in terms of the higher-rank numerical range of the operator $(T \otimes id_n)(C_{Q})$ \cite{CKZ06}. In particular, it implies that if $Q$ has constant gate fidelity, then $(T \otimes id_n)(C_{Q})$ must have non-empty rank-$[n(n+1)/2]$ numerical range. Because the higher-rank numerical range of a Hermitian operator is well-understood in terms of its eigenvalues, it follows that the middle $n$ eigenvalues of $(T \otimes id_n)(C_{Q})$ must all be equal in order for $Q$ to have constant gate fidelity.

\section{Gate Fidelity in Small Dimensions}\label{sec:SmallDims}

Note that for any quantum channel $\cl{E}$ and state $\ket{\phi} \in \cl{H}_n$ it is trivially the case that $\Tr(\cl{E}(\ketbra{\phi}{\phi})\ketbra{\phi}{\phi}) = \Tr(\cl{E}^\dagger(\ketbra{\phi}{\phi})\ketbra{\phi}{\phi})$, so $\cl{F}_{\cl{E}} = \cl{F}_{\cl{E}^\dagger}$. Similarly, if $\cl{Q}$ is a quantum channel, $r \geq 0$, and $\cl{E}$ is a quantum channel, then $\cl{Q} + r(\cl{E} - \cl{E}^\dagger)$ also has gate fidelity equal to that of $\cl{Q}$. A particular consequence of this observation is that for any quantum channel $\cl{Q}$ with full-rank Choi matrix, there is another quantum channel (not equal to $\cl{Q}^\dagger$) with the same gate fidelity -- a fact that was originally proved in dimensions $n \geq 4$ in \cite{Mag11}. To construct such a map it suffices to pick a unital channel $\cl{E}$ and then choose a sufficiently small $r > 0$.

A natural question to ask is whether or not the converse of the observation made in the previous paragraph holds. That is, if two quantum channels $\cl{Q}$ and $\cl{R}$ have the same gate fidelity, do they differ by some $r(\cl{E} - \cl{E}^\dagger)$? One simplification we can make right away is that we can assume without loss of generality that $\cl{E}$ is unital. Indeed, it is easily-verified that the channel $\cl{E}$ is unital if and only if $\cl{R} = \cl{Q} + r(\cl{E} - \cl{E}^\dagger)$ is trace-preserving. The example used in the construction of \cite[Theorem 1]{Mag11} shows that such a channel $\cl{E}$ need not exist when $n \geq 4$, even if $\cl{Q},\cl{R}$ are assumed to be unital. We now present an example to show that $\cl{E}$ need not exist when $n = 3$, as long as we do not require that $\cl{Q}$ and $\cl{R}$ be unital.

\begin{exam}\label{exam:n3counter}{\rm
	Consider the two channels $\cl{Q},\cl{R} : \cl{L}(\cl{H}_3) \rightarrow \cl{L}(\cl{H}_3)$ defined in the standard basis by the following Choi matrices:
	\begin{align*}
		C_{\cl{Q}} := \frac{1}{4}\begin{bmatrix}2 & 0 & 0 & 0 & 0 & 0 & 0 & 0 & 1 \\
0 & 1 & 0 & 0 & 0 & 0 & 0 & 0 & 0 \\
0 & 0 & 1 & 0 & 0 & 0 & 0 & 0 & 0 \\
0 & 0 & 0 & 1 & 0 & 0 & 0 & 0 & 0 \\
0 & 0 & 0 & 0 & 2 & 0 & 0 & 0 & 0 \\
0 & 0 & 0 & 0 & 0 & 1 & 0 & 0 & 0 \\
0 & 0 & 0 & 0 & 0 & 0 & 1 & 0 & 0 \\
0 & 0 & 0 & 0 & 0 & 0 & 0 & 1 & 0 \\
1 & 0 & 0 & 0 & 0 & 0 & 0 & 0 & 2\end{bmatrix}, \ \ C_{\cl{R}} := \frac{1}{4}\begin{bmatrix}2 & 0 & 0 & 0 & 0 & 0 & 0 & 0 & 0 \\
0 & 0 & 0 & 0 & 0 & 0 & 0 & 0 & 0 \\
0 & 0 & 2 & 0 & 0 & 0 & 0 & 0 & 0 \\
0 & 0 & 0 & 2 & 0 & 0 & 0 & 0 & 0 \\
0 & 0 & 0 & 0 & 2 & 0 & 0 & 0 & 1 \\
0 & 0 & 0 & 0 & 0 & 0 & 0 & 0 & 0 \\
0 & 0 & 0 & 0 & 0 & 0 & 2 & 0 & 0 \\
0 & 0 & 0 & 0 & 0 & 0 & 0 & 0 & 0 \\
0 & 0 & 0 & 0 & 1 & 0 & 0 & 0 & 2\end{bmatrix}.
	\end{align*}
	It is easily verified that $C_{\cl{Q}}$ and $C_{\cl{R}}$ are both positive semidefinite and $\Tr_2(C_{\cl{Q}}) = \Tr_2(C_{\cl{R}}) = I$, so $\cl{Q}$ and $\cl{R}$ are both indeed quantum channels (but $\Tr_1(C_{\cl{Q}}) = I \neq \Tr_1(C_{\cl{R}})$ so $\cl{Q}$ is unital but $\cl{R}$ is not). It is also easily verified that $P_{\cl{S}}(T \otimes id_3)(C_{\cl{Q}} - C_{\cl{R}})P_{\cl{S}} = 0$, so $\cl{F}_{\cl{Q}} = \cl{F}_{\cl{R}}$ by Theorem~\ref{thm:equal_gate_choi}.
	
	On the other hand, because $\cl{Q}$ is unital, it follows that if $\cl{E}$ is a unital quantum channel then $\cl{Q}(I) + r(\cl{E}(I) - \cl{E}^{\dagger}(I)) = I + r(I - I) = I$, so $\cl{Q} + r(\cl{E} - \cl{E}^{\dagger})$ is unital as well. Thus there does not exist $r \geq 0$ and a unital quantum channel $\cl{E}$ such that $\cl{R} = \cl{Q} + r(\cl{E} - \cl{E}^{\dagger})$.
}\end{exam}

We now present the main result of this section, which shows that when $n = 2$, the converse of our previous observation does indeed hold. That is, if two qubit channels have the same gate fidelity then their difference equals the scaled difference of some pair of dual channels. We also show that this converse still holds when $n = 3$, in spite of Example~\ref{exam:n3counter}, under some slightly stronger hypotheses.

\begin{thm}\label{thm:small_dims}
	Let $\cl{Q},\cl{R} : \cl{L}(\cl{H}_n) \rightarrow \cl{L}(\cl{H}_n)$ be quantum channels. Suppose that either $n = 2$, or $n = 3$ and $\cl{Q}(I) = \cl{R}(I)$. Then $\cl{F}_{\cl{Q}} = \cl{F}_{\cl{R}}$ if and only if there exists $r \geq 0$ and a unital quantum channel $\cl{E} : \cl{L}(\cl{H}_n) \rightarrow \cl{L}(\cl{H}_n)$ such that $\cl{R} = \cl{Q} + r(\cl{E} - \cl{E}^\dagger)$.
\end{thm}

\begin{proof}

	As has already been discussed, the ``if'' direction clearly holds in any dimension. To see the ``only if'' direction, we first consider the case when $n = 2$. Let's choose an orthogonal basis of (unnormalized) symmetric and antisymmetric states:
	\begin{align*}
		\ket{a} := \ket{01} - \ket{10}, \quad \ket{s_1} := \ket{01} + \ket{10}, \quad \ket{s_2} := \ket{00}, \quad \ket{s_3} := \ket{11}.
	\end{align*}
	From Theorem~\ref{thm:equal_gate_choi} we know $P_{\cl{S}}(T \otimes id_2)(C_{\cl{\cl{Q}}} - C_{\cl{R}})P_{\cl{S}} = 0$. Thus there exist $\alpha \in \bb{R}$ and $c_1,c_2,c_3 \in \bb{C}$ such that we can write
	\begin{align*}
		(T \otimes id_2)(C_{\cl{Q}} - C_{\cl{R}}) = \alpha \ketbra{a}{a} + \sum_{j=1}^3 \big(c_{j}\ketbra{a}{s_j} + \overline{c_{j}}\ketbra{s_j}{a}\big).
	\end{align*}
	The fact that $\cl{\cl{Q}}$ and $\cl{R}$ are trace-preserving implies that $\Tr_2(C_{\cl{Q}}) = \Tr_2(C_{\cl{R}}) = I$, which implies that $\alpha = 0$, $c_1$ is purely imaginary, and $c_3 = \overline{c_2}$. It follows that we can write $(T \otimes id_2)(C_{\cl{Q}} - C_{\cl{R}})$ in the standard basis as
	\begin{align}\label{eq:2by2_qr}
		(T \otimes id_2)(C_{\cl{Q}} - C_{\cl{R}}) = \begin{bmatrix}0 & \overline{c_2} & -\overline{c_2} & 0 \\ c_2 & 0 & 2c_1 & \overline{c_2} \\ -c_2 & 2\overline{c_1} & 0 & -\overline{c_2} \\ 0 & c_2 & -c_2 & 0\end{bmatrix}.
	\end{align}
	Now choose $r \geq 2|c_2| + 2|c_1|$ and define a channel $\cl{E} : \cl{L}(\cl{H}_2) \rightarrow \cl{L}(\cl{H}_2)$ via the Choi matrix
	\begin{align*}
		C_{\cl{E}} = \frac{1}{2r}\begin{bmatrix}r & 0 & -2c_2 & 2\overline{c_1} \\ 0 & r & 0 & 2c_2 \\ -2\overline{c_2} & 0 & r & 0 \\ 2c_1 & 2\overline{c_2} & 0 & r\end{bmatrix}.
	\end{align*}
	It is easily verified that $\Tr_1(C_{\cl{E}}) = \Tr_2(C_{\cl{E}}) = I$ so $\cl{E}$ is unital and trace-preserving, and the fact that $\cl{E}$ is completely positive follows from the fact that its Choi matrix is diagonally dominant and hence positive semidefinite. It is a simple calculation using Lemma~\ref{lem:dual_choi} to verify that
	\begin{align}\label{eq:2by2_ee}
		r(T \otimes id_2)(C_{\cl{E}} - C_{\cl{E}^{\dagger}}) = \begin{bmatrix}0 & \overline{c_2} & -\overline{c_2} & 0 \\ c_2 & 0 & 2c_1 & \overline{c_2} \\ -c_2 & 2\overline{c_1} & 0 & -\overline{c_2} \\ 0 & c_2 & -c_2 & 0\end{bmatrix}.
	\end{align}
	By comparing Equations~\eqref{eq:2by2_qr} and~\eqref{eq:2by2_ee} we see that $\cl{Q} - \cl{R} = r(\cl{E} - \cl{E}^{\dagger})$, which completes the proof for the case when $n = 2$.
	
	The case when $n = 3$ and $\cl{Q}(I) = \cl{R}(I)$ is proved analogously, but the algebra is more involved. Choose an orthogonal basis of (unnormalized) symmetric and antisymmetric states:
	\begin{align*}
		\ket{a_1} := \ket{01} - \ket{10}, \quad & \ket{s_1} := \ket{01} + \ket{10} \\
		\ket{a_2} := \ket{02} - \ket{20}, \quad & \ket{s_2} := \ket{02} + \ket{20} \\
		\ket{a_3} := \ket{12} - \ket{21}, \quad & \ket{s_3} := \ket{12} + \ket{21} \\
		& \ket{s_4} := \ket{00} \\
		& \ket{s_5} := \ket{11} \\
		& \ket{s_6} := \ket{22}
	\end{align*}
	From Theorem~\ref{thm:equal_gate_choi} we know $P_{\cl{S}}(T \otimes id_3)(C_{\cl{\cl{Q}}} - C_{\cl{R}})P_{\cl{S}} = 0$. Thus there exist $\{\alpha_j\} \in \bb{R}$ and $\{c_{j,k}\} \in \bb{C}$ ($1 \leq j \leq 3, 1 \leq k \leq 6$) such that we can write
	\begin{align*}
		(T \otimes id_3)(C_{\cl{Q}} - C_{\cl{R}}) = \sum_{j=1}^3 \alpha_j \ketbra{a_j}{a_j} + \sum_{j=1}^3\sum_{k=1}^6 \big(c_{j,k}\ketbra{a_j}{s_k} + \overline{c_{j,k}}\ketbra{s_k}{a_j}\big).
	\end{align*}
	The facts that $\cl{Q}$ and $\cl{R}$ are trace-preserving and $\cl{Q}(I) = \cl{R}(I)$ imply that $\Tr_1((T \otimes id_3)(C_{\cl{Q}} - C_{\cl{R}})) = \Tr_2((T \otimes id_3)(C_{\cl{Q}} - C_{\cl{R}})) = 0$, which implies $\alpha_1 = \alpha_2 = \alpha_3 = 0$. The partial trace conditions also imply the following conditions:
	\begin{align*}
		{\rm Re}(c_{1,1}) & = {\rm Re}(c_{3,3}) = -{\rm Re}(c_{2,2}) \\
		c_{2,3} & = \overline{c_{1,4}} - c_{1,5} - \overline{c_{3,2}} \\
		c_{2,6} & = \overline{c_{2,4}} + \overline{c_{3,1}} - c_{1,3} \\
		c_{3,6} & = c_{1,2} + \overline{c_{2,1}} + \overline{c_{3,5}}
	\end{align*}
	If $c_{1,1} = b + id_1$, ${\rm Im}(c_{2,2}) = d_2$ and ${\rm Im}(c_{3,3}) = d_3$ then we can write $(T \otimes id_3)(C_{\cl{Q}} - C_{\cl{R}})$ in the standard basis as
	\begin{align}\label{eq:3by3_qr}
		\begin{bmatrix}
			0 & \overline{c_{1,4}} & \overline{c_{2,4}} & -\overline{c_{1,4}} & 0 & \overline{c_{3,4}} & -\overline{c_{2,4}} & -\overline{c_{3,4}} & 0 \\
			* & 2b & c_{1,2} + \overline{c_{2,1}} & 2id_1 & c_{1,5} & c_{1,3} + \overline{c_{3,1}} & c_{1,2} - \overline{c_{2,1}} & c_{1,3} - \overline{c_{3,1}} & c_{1,6} \\
			* & * & -2b & -\overline{c_{1,2}} + c_{2,1} & c_{2,5} & \overline{c_{1,4}} - c_{1,5} & 2id_2 & c_{2,3}-\overline{c_{3,2}} & c_{2,6} \\
			* & * & * & -2b & -c_{1,5} & \overline{c_{3,1}} - c_{1,3} & -\overline{c_{2,1}} - c_{1,2} & -\overline{c_{3,1}} - c_{1,3} & -c_{1,6} \\
			* & * & * & * & 0 & \overline{c_{3,5}} & -\overline{c_{2,5}} & -\overline{c_{3,5}} & 0 \\
			* & * & * & * & * & 2b & c_{3,2} - \overline{c_{2,3}} & 2id_3 & c_{3,6} \\
			* & * & * & * & * & * & 2b & -\overline{c_{1,4}} + c_{1,5} & -c_{2,6} \\
			* & * & * & * & * & * & * & -2b & -c_{3,6} \\
			* & * & * & * & * & * & * & * & 0\end{bmatrix},
	\end{align}
	where the $*$ in the $(j,k)$-entry indicates the complex conjugate of the $(k,j)$-entry. Now let $r > 0$ be arbitrarily large, $\varepsilon > 0$ and $0 < s,t < 1 - \varepsilon$, and define $u := 1 - \varepsilon - s$ and $v := 1 - \varepsilon - t$. We shall consider the linear map $\cl{E}$ defined by the Choi matrix
	\begin{align*}
		C_{\cl{E}} = \frac{1}{3r}\begin{bmatrix}
			3r\varepsilon & 0 & 0 & -3c_{1,4} & -3id_1 & -3c_{1,2} & -3c_{2,4} & -3c_{2,1} & -3id_2 \\
			* & 3rs & 3c_{1,2} & 0 & 0 & 0 & -3c_{3,4} & -3c_{3,1} & 3c_{1,4} - 6c_{3,2} \\
			* & * & 3ru & 0 & 3\overline{c_{1,3}} & 3c_{1,4} & 0 & 0 & 3c_{2,4} + 3c_{3,1} \\
			* & * & * & 3rt & -3c_{1,5} & -3c_{1,3} & -3c_{2,1} & -3c_{2,5} & 3c_{1,5} \\
			* & * & * & * & 3r\varepsilon & 0 & -3c_{3,1} & -3c_{3,5} & -3id_3 \\
			* & * & * & * & * & 3rv & -3\overline{c_{1,6}} & 0 & 3c_{2,1} + 3c_{3,5} \\
			* & * & * & * & * & * & 3rv & 3c_{1,5} & 3c_{1,3} \\
			* & * & * & * & * & * & * & 3ru & -3c_{1,2} \\
			* & * & * & * & * & * & * & * & 3r(1 - u - v)\end{bmatrix}.
	\end{align*}
	In particular, choose $s$ and $t$ so that $s+t \geq 1$ and $s-t = 2b/r$. This is always possible because we can choose $s$ so that $2b/r < s < 1$, set $t = s-2b/r$, and then choose some sufficiently small $\varepsilon > 0$. It follows that each diagonal entry of $C_{\cl{E}}$ is strictly positive, and so by choosing $r$ sufficiently large it becomes diagonally dominant and hence positive semidefinite. It is worth noting that, although $t$ does depend on $r$, increasing $r$ only increases $t$ so diagonal dominance won't be interfered with as $t$ varies. It is easily verified that $\Tr_1(C_{\cl{E}}) = \Tr_2(C_{\cl{E}}) = I$ so $\cl{E}$ is in fact a unital quantum channel.
	
	Furthermore, it is a simple (albeit tedious) calculation using Lemma~\ref{lem:dual_choi} to verify that $r(T \otimes id_2)(C_{\cl{E}} - C_{\cl{E}^{\dagger}})$, when written in the standard basis, is exactly the matrix~\eqref{eq:3by3_qr}. It follows that $\cl{Q} - \cl{R} = r(\cl{E} - \cl{E}^{\dagger})$, which completes the proof.
\end{proof}

It is worth pointing out where the techniques used in the proof of Theorem~\ref{thm:small_dims} break down in the $n \geq 4$ case. The proofs in the $n = 2$ and $n = 3$ cases both begin by making use of the constraint $\Tr_2((T \otimes id_n)(C_{\cl{Q}} - C_{\cl{R}})) = 0$ (and similarly for $\Tr_1$ when $n = 3$) to restrict the $n(n-1)/2$ real and $n^2(n-1)(n+1)/4$ complex coefficients that define $(T \otimes id_n)(C_{\cl{Q}} - C_{\cl{R}})$. In particular, we need the partial trace restrictions to imply that $P_{\cl{A}}(T \otimes id_n)(C_{\cl{Q}} - C_{\cl{R}})P_{\cl{A}} = 0$, where $P_{\cl{A}} := I - P_{\cl{S}}$ is the projection onto the antisymmetric subspace. In the $n \geq 4$ case the two partial trace constraints aren't strong enough to guarantee this property holds, as demonstrated by the map in the proof of \cite[Theorem 1]{Mag11}.

\section{Minimum Gate Fidelity}\label{sec:minGateFid}

Calculating the minimum gate fidelity is a problem that is expected to be difficult, and few general methods of approximating and bounding it are known. Our main result of this section shows that the minimal fidelity of a channel can be written in terms of the $S(1)$-norm \cite{JK10a,JK10b}, defined on positive semidefinite operators as
\begin{align*}
	\big\|X\big\|_{S(1)} := \sup_{\ket{\phi},\ket{\psi}}\big\{\bra{\phi\psi}X\ket{\phi\psi} \big\}.
\end{align*}
Our reason for expressing the minimum gate fidelity in terms of this norm rather than in terms of other related minimizations or maximizations is that several bounds, inequalities and properties of the $S(1)$-norm are already known \cite{JK10a,JK10b,JKPP10,J10}, whereas other similar supremums or infimums appear to be more nebulous. In particular, computation of the $S(1)$-norm on positive semidefinite operators is equivalent to the problem of determining whether or not a given operator is an entanglement witness \cite[Corollary 4.9]{JK10a}, which is equivalent to the problem of determining whether or not a superoperator is positive. The problems of characterizing positive superoperators \cite{TT83,BFP04} and entanglement witnesses and separable states \cite{P96,HHH96,B02,DPS04,CKo09,HHHH09} have been studied extensively, so an abundance of results from operator theory and entanglement theory now apply in this setting.

All separability criteria that have been developed over the past several years now automatically translate into methods of bounding minimum gate fidelity. As a particularly important example of this, symmetric extensions \cite{DPS04} can now be used to compute minimum gate fidelity within any desired accuracy simply by performing the optimization over separable states over the set of states with $k$-symmetric extension instead. The optimization over states with $k$-symmetric extension is a semidefinite program \cite{VB96}, which can be computed efficiently \cite{GLS93}, and precise bounds for how far away the optimal value of the $k$th semidefinite program is from the optimization over separable states are given in \cite{NOP09}.

\begin{thm}\label{thm:min_fid}
	Let $\cl{E} : \cl{L}(\cl{H}_n) \rightarrow \cl{L}(\cl{H}_n)$ be a quantum channel and let $\lambda_1$ be the maximal eigenvalue of $P_{\cl{S}}(T \otimes id_n)(C_{\cl{E}})P_{\cl{S}}$. Then
	\begin{align*}
		\cl{F}_{\cl{E}}^{min} = \lambda_1 - \big\| \lambda_1 P_{\cl{S}} - P_{\cl{S}}(T \otimes id_n)(C_{\cl{E}})P_{\cl{S}} \big\|_{S(1)}.
	\end{align*}
\end{thm}

\begin{proof}
	Using Lemma~\ref{lem:fidel_choi} with $\Lambda := \cl{E}$ reveals that
	\begin{align}\label{eq:min_formula}
		\cl{F}_{\cl{E}}^{min} = \min_{\ket{\phi} \in \cl{H}_n}\big\{ \bra{\phi\phi}(T \otimes id_n)(C_{\cl{E}})\ket{\phi\phi} \big\}  = \lambda_1 - \max_{\ket{\phi} \in \cl{H}_n}\big\{ \bra{\phi\phi}(T \otimes id_n)(\lambda_1 I - C_{\cl{E}})\ket{\phi\phi} \big\}.
	\end{align}
	For convenience, we will now define $X := P_{\cl{S}}(T \otimes id_n)(\lambda_1 I - C_{\cl{E}})P_{\cl{S}}$. Notice that $X$ is positive semidefinite. We now note that
	\begin{align*}
		\max_{\ket{\phi} \in \cl{H}_n}\big\{ \bra{\phi\phi}(T \otimes id_n)(\lambda_1 I - C_{\cl{E}})\ket{\phi\phi} \big\} \leq \max_{\ket{\psi},\ket{\chi} \in \cl{H}_n}\big\{ \bra{\psi\chi}X\ket{\psi\chi} \big\}
	\end{align*}
	trivially by letting $\ket{\psi} = \ket{\chi} = \ket{\phi}$. To see that the opposite inequality holds as well (and hence complete the proof), suppose $\ket{\psi} \neq \ket{\chi}$ and observe that $P_{\cl{S}}\ket{\psi\chi} = \frac{1}{2}(\ket{\psi\chi} + \ket{\chi\psi})$ is a scalar multiple of a symmetric state with Schmidt rank $2$. It follows via the Takagi factorization that we can write $P_{\cl{S}}\ket{\psi\chi} = \alpha\ket{\rho\rho} + \beta\ket{\sigma\sigma}$ for some $\ket{\rho},\ket{\sigma} \in \cl{H}_n$ and $\alpha,\beta \geq 0$. Suppose without loss of generality that
	\begin{align*}
		\bra{\rho\rho}X\ket{\rho\rho} \geq \bra{\sigma\sigma}X\ket{\sigma\sigma}.
	\end{align*}
	
	Then write $X$ in its Spectral Decomposition as $X = \sum_i \lambda_i \ketbra{v_i}{v_i}$ and define the $i^{th}$ component of two vectors $\rho^\prime$ and $\sigma^\prime$ by $\rho_i^\prime := \sqrt{\lambda_i}|\braket{v_i}{\rho\rho}|$ and $\sigma_i^\prime := \sqrt{\lambda_i}|\braket{\sigma\sigma}{v_i}|$. Applying the Cauchy-Schwarz inequality to $\rho^\prime$ and $\sigma^\prime$ shows
	\begin{align*}
		|\bra{\sigma\sigma}X\ket{\rho\rho}| \leq \sqrt{\bra{\rho\rho}X\ket{\rho\rho}}\sqrt{\bra{\sigma\sigma}X\ket{\sigma\sigma}} \leq \bra{\rho\rho}X\ket{\rho\rho}.
	\end{align*}
	Putting all of this together shows that
	\begin{align*}
		\bra{\psi\chi}X\ket{\psi\chi} & = (\alpha\bra{\rho\rho} + \beta\bra{\sigma\sigma})X(\alpha\ket{\rho\rho} + \beta\ket{\sigma\sigma}) \\
		& = \alpha^2\bra{\rho\rho}X\ket{\rho\rho} + \alpha\beta(\bra{\rho\rho}X\ket{\sigma\sigma} + \bra{\sigma\sigma}X\ket{\rho\rho}) + \beta^2\bra{\sigma\sigma}X\ket{\sigma\sigma} \\
		& \leq (\alpha^2 + \beta^2)\bra{\rho\rho}X\ket{\rho\rho} + \alpha\beta(|\bra{\rho\rho}X\ket{\sigma\sigma}| + |\bra{\sigma\sigma}X\ket{\rho\rho}|) \\
		& \leq (\alpha^2 + 2\alpha\beta + \beta^2)\bra{\rho\rho}X\ket{\rho\rho} \\
		& = (\alpha + \beta)^2\bra{\rho\rho}X\ket{\rho\rho}.
	\end{align*}
	
	Thus, if we can prove that $\alpha + \beta \leq 1$ then we are done. To this end, first note that without loss of generality we can assume that $\braket{\psi}{\chi}$ is real, simply by adjusting the global phase between $\ket{\psi}$ and $\ket{\chi}$ appropriately. Now recall from the Takagi factorization that $\alpha$ and $\beta$ are the square roots of the eigenvalues of the matrix
	\begin{align*}
		AA^* & := \frac{1}{4}\big(\ket{\psi}\overline{\bra{\chi}} + \ket{\chi}\overline{\bra{\psi}}\big)\big(\overline{\ket{\chi}}\bra{\psi} + \overline{\ket{\psi}}\bra{\chi}\big) \\
		& = \frac{1}{4}\big(\ketbra{\psi}{\psi} + \braket{\psi}{\chi} (\ketbra{\chi}{\psi} + \ketbra{\psi}{\chi}) + \ketbra{\chi}{\chi}\big).
	\end{align*}
	It is easily verified that eigenvectors of $AA^*$ are $\ket{\psi} \pm \ket{\chi}$ and the associated eigenvalues are
	\begin{align*}
		\frac{1}{4}\big( \braket{\psi}{\chi}^2 \pm 2\braket{\psi}{\chi} + 1 \big).
	\end{align*}
	If we add the square roots of these eigenvalues, we get
	\begin{align*}
		\alpha + \beta & = \frac{1}{2}\sqrt{ \braket{\psi}{\chi}^2 + 2\braket{\psi}{\chi} + 1 } + \frac{1}{2}\sqrt{ \braket{\psi}{\chi}^2 - 2\braket{\psi}{\chi} + 1 } \\
		& = \frac{1}{2}\sqrt{ (\braket{\psi}{\chi} + 1)^2 } + \frac{1}{2}\sqrt{ (\braket{\psi}{\chi} - 1)^2 } \\
		& = \frac{1}{2}\big|1 + \braket{\psi}{\chi}\big| + \frac{1}{2}\big|1 - \braket{\psi}{\chi}\big| \\
		& = 1,
	\end{align*}
	where the final equality follows from the fact that $-1 \leq \braket{\psi}{\chi} \leq 1$. The result follows.
\end{proof}

Several bounds and results on the minimum gate fidelity follow immediately from the corresponding results on the $S(1)$-norm derived in \cite{JK10a,JK10b}. We present a brief selection of these results here for completeness.

\begin{cor}
	Let $\cl{E} : \cl{L}(\cl{H}_n) \rightarrow \cl{L}(\cl{H}_n)$ be a quantum channel. Denote the eigenvalues of $P_{\cl{S}}(T \otimes id_n)(C_{\cl{E}})P_{\cl{S}}$ supported on $P_{\cl{S}}$ by $\lambda_1 \geq \lambda_2 \geq \cdots \geq \lambda_{n(n+1)/2}$ (i.e., these are the eigenvalues of $P_{\cl{S}}(T \otimes id_n)(C_{\cl{E}})P_{\cl{S}}$ with $n(n-1)/2$ zero eigenvalues removed). Let $\alpha_j$ be the maximal Schmidt coefficient of the eigenvector corresponding to $\lambda_j$. Then
	\begin{align*}
		\max_{j}\{ (\lambda_1 - \lambda_j) \alpha_j^2 \} \leq \lambda_1 - \cl{F}_{\cl{E}}^{min} \leq \min\big\{\lambda_1 - \lambda_{n(n+1)/2}, \sum_j (\lambda_1 - \lambda_j) \alpha_j^2\big\}.
	\end{align*}
\end{cor}

\begin{proof}
	The fact that $\lambda_1 - \cl{F}_{\cl{E}}^{min} \leq \lambda_1 - \lambda_{n(n+1)/2}$ follows immediately from Theorem~\ref{thm:min_fid} and the fact that $\|\cdot\|_{S(1)} \leq \|\cdot\|$. The other upper bound of $\lambda_1 - \cl{F}_{\cl{E}}^{min}$ follows from \cite[Theorem 3.3 and Proposition 4.11]{JK10a}. The lower bound can be derived by using the spectral decomposition to write
	\begin{align*}
		P_{\cl{S}}(T \otimes id_n)(C_{\cl{E}})P_{\cl{S}} = \sum_{j} \lambda_j \ketbra{v_j}{v_j}.
	\end{align*}
	If $\ket{v} \in \cl{H}_n \otimes \cl{H}_n$ is the separable state corresponding to the maximal Schmidt coefficient $\alpha_j$ of $\ket{v_j}$ then
	\begin{align*}
		\bra{v}P_{\cl{S}}(T \otimes id_n)(\lambda_1I - C_{\cl{E}})P_{\cl{S}}\ket{v} & = \sum_i (\lambda_1 - \lambda_i) |\braket{v_i}{v}|^2 \\
		& = (\lambda_1 - \lambda_j) \alpha_j^2 + \sum_{i\neq j} (\lambda_1 - \lambda_i) |\braket{v_i}{v}|^2 \\
		& \geq (\lambda_1 - \lambda_j) \alpha_j^2.
	\end{align*}
	The corresponding lower bound follows by letting $j$ range from $1$ to $n(n+1)/2$.
\end{proof}

\begin{remark}{\rm
In the case that $n = 2$, the $S(1)$-norm can be efficiently computed to any desired accuracy via semidefinite programming \cite{JK10b}. As a corollary of this fact, we now have a semidefinite program for efficiently computing $\cl{F}_{\cl{E}}^{min}$ of qubit channels $\cl{E} : \cl{L}(\cl{H}_2) \rightarrow \cl{L}(\cl{H}_2)$ to any desired accuracy. The primal and dual forms of the semidefinite program in question are as follows:
\begin{align*}\label{sp:minGateFid}
\begin{matrix}
\begin{tabular}{r l}
\multicolumn{2}{c}{{\bf Primal problem}} \\
\text{minimize:} & $\lambda_1 - \Tr\big((\lambda_1 P_{\cl{S}} - P_{\cl{S}}(T \otimes id_2)(C_{\cl{E}})P_{\cl{S}})\rho\big)$ \\
\text{subject to:} & $\rho \geq 0, (id_2 \otimes T)(\rho) \geq 0$ \\
\ & $\Tr(\rho) = 1$ \\
\ & \ \\
\multicolumn{2}{c}{{\bf Dual problem}} \\
\text{maximize:} & $\lambda_1 - \big\|(T \otimes id_2)(Y) + (P_{\cl{S}}(T \otimes id_2)(\lambda_1 I - C_{\cl{E}})P_{\cl{S}})\big\|$ \\
\text{subject to:} & $Y \geq 0$
\end{tabular}
\end{matrix}
\end{align*}

Using the MATLAB code provided in \cite{JK10b} to solve this semidefinite program, we are able to approximate the distribution of the minimum gate fidelity when $n = 2$. Figure~\ref{fig:min_fid} shows the distribution of $\cl{F}_{\cl{E}}^{min}$ and $\overline{\cl{F}_{\cl{E}}}$ when the quantum channel $\cl{E}$ is chosen by picking a Haar-uniform unitary $U \in \cl{L}(\cl{H}_4) \otimes \cl{L}(\cl{H}_2)$ and then setting $\cl{E}(\rho) \equiv \Tr_1(U (\ketbra{0}{0} \otimes \rho) U^{\dagger})$.
}\end{remark}

\begin{figure}[ht]
\begin{center}
\includegraphics[width=0.9\textwidth]{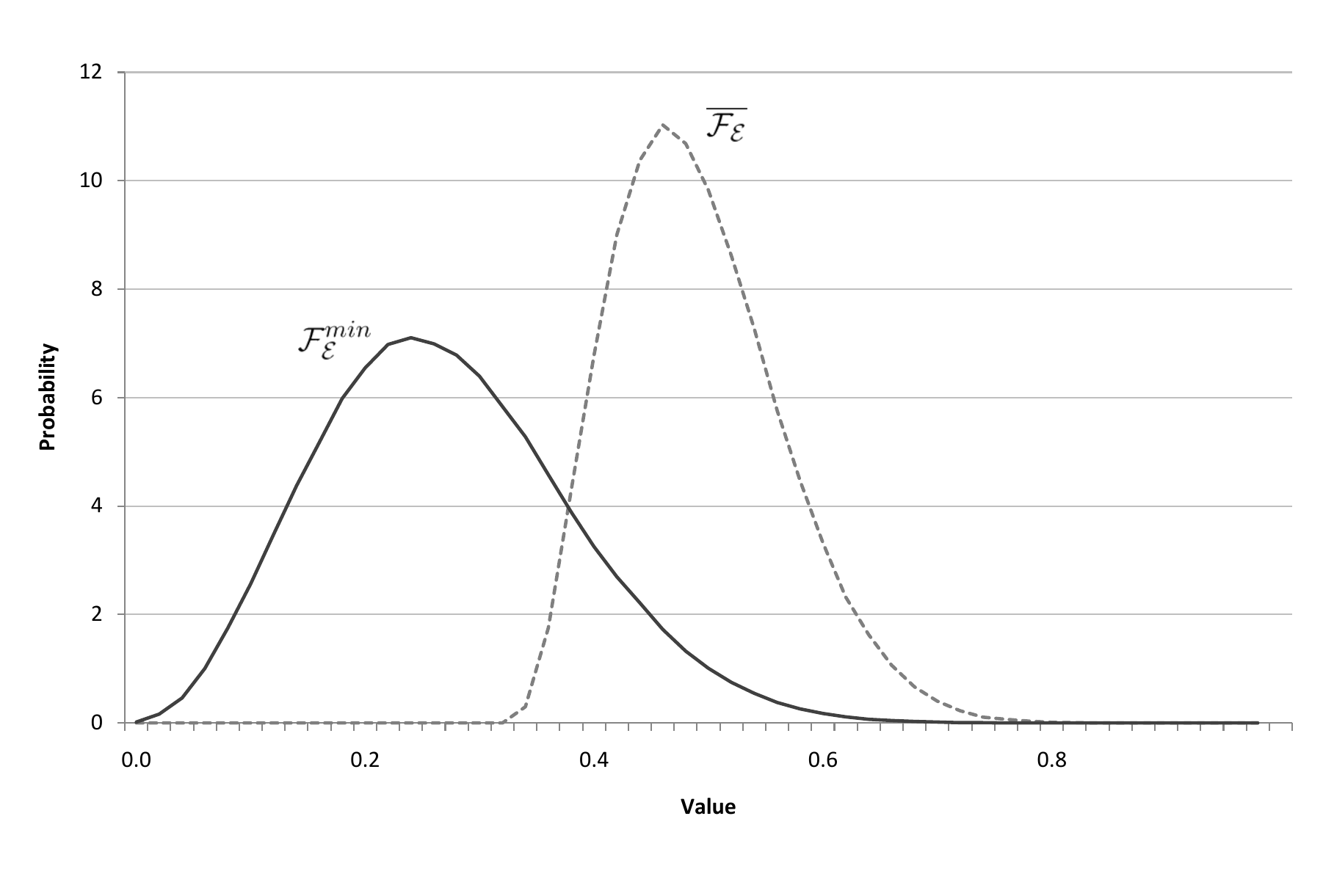}
\end{center}
\caption{Approximate distributions of $\cl{F}_{\cl{E}}^{min}$ and $\overline{\cl{F}_{\cl{E}}}$ when $n = 2$, based on $5 \cdot 10^5$ randomly-generated qubit channels.}\label{fig:min_fid}
\end{figure}

\section{Outlook}\label{sec:outlook}

We have seen that the quantum gate fidelity of a quantum channel $\cl{E}$ is characterized by a particular operator, $P_{\cl{S}}(T \otimes id_n)(C_{\cl{E}})P_{\cl{S}}$. Furthermore, the average and minimum gate fidelities can be expressed in terms of this operator relatively simply, as can the variance of the gate fidelity.

Although the $S(1)$-norm in general is NP-HARD to compute, this does not immediately imply that computing $\cl{F}_{\cl{E}}^{min}$ is NP-HARD as well, because the operator whose norm is computed in Theorem~\ref{thm:min_fid} has a special form (in particular, it is supported on the symmetric subspace). Determining whether or not $\cl{F}_{\cl{E}}^{min}$ is difficult to compute would be a great step toward a better understanding of gate fidelities. Also, we have provided a semidefinite program that computes $\cl{F}_{\cl{E}}^{min}$ for any qubit channel $\cl{E}$, but it might be possible to do better than this and find an explicit formula for the minimum gate fidelity in this situation.

\vspace{0.1in}

\noindent{\bf Acknowledgements.} We are grateful to the referees for helpful comments. Thanks are extended to Moritz Ernst for pointing out an error in an early version of Proposition~\ref{prop:avg_gate_fid}. N.J. was supported by an NSERC Canada Graduate Scholarship and the University of Guelph Brock Scholarship. D.W.K. was supported by Ontario Early Researcher Award 048142, NSERC Discovery Grant 400160 and NSERC Discovery Accelerator Supplement 400233.


\end{document}